\documentclass[aps,pra,onecolumn,10pt,longbibliography,amsmath,amssymb]{revtex4-2}
\usepackage[T1]{fontenc}
\usepackage{lmodern,microtype}
\usepackage{amsthm,mathtools}
\usepackage{graphicx,booktabs,array,enumitem,xcolor,placeins}
\definecolor{linkblue}{HTML}{17324D}
\usepackage[colorlinks=true,allcolors=linkblue]{hyperref}
\hypersetup{pdftitle={Simultaneous Circuit Tests for Finite Reversible Models of Quantum Control},pdfauthor={Josef Bruzzese}}
\usepackage{tikz}
\usetikzlibrary{arrows.meta,positioning}
\setlist{nosep,leftmargin=*}
\newtheorem{theorem}{Theorem}[section]
\newtheorem{proposition}[theorem]{Proposition}
\newtheorem{lemma}[theorem]{Lemma}

\theoremstyle{definition}
\theoremstyle{remark}
\newcommand{\ket}[1]{\lvert #1\rangle}
\newcommand{\bra}[1]{\langle #1\rvert}
\newcommand{\tr}{\operatorname{tr}}
\newcommand{\diag}{\operatorname{diag}}
\newcommand{\op}{\mathrm{op}}
\newcommand{\id}{\mathrm{id}}
\newcommand{\W}{\mathcal W}
\newcommand{\Q}{\mathcal Q}
\newcommand{\Pp}{\mathcal P}
\newcommand{\C}{\mathcal C}
\newcommand{\eps}{\varepsilon}

\begin{document}

\title{Simultaneous Circuit Tests for Finite Reversible\texorpdfstring{\\}{ }
Models of Quantum Control}
\author{Josef Bruzzese}
\noaffiliation
\date{3 October 2026}

\begin{abstract}
A finite reversible model can approximate individual quantum gates while retaining freedom in its transition tables and preparation contexts. We formulate the simultaneous certification problem in which one stationary controller and one preparation law must reproduce an entire collection of circuit probabilities. For a specified frame catalogue, accuracy graph and calibration constraints, we give an exact shared-table feasibility formulation and distinguish certificates for a fixed table from certificates for the complete model class. A four-frame qubit example exhibits a strict separation between fitting two circuits individually and fitting them jointly. A larger exact benchmark uses 464 frames and a phase symmetry that removes preparation-prior dependence. We independently reconstruct its admissible graph and certify a sharp repeated-circuit horizon: one controller stays within probability error $1/3$ through 192 blocks, whereas every allowed controller violates that threshold by block 193. A separate controller reproduces all 8,190 nonempty Hadamard--phase words of length at most twelve within the same tolerance. We also derive a finite-window spectral obstruction and finite-shot confidence tests, illustrated by a reproducible sampling experiment. The conclusions concern the stated finite reversible model classes; they do not establish a universal memory bound or a discrepancy for physical quantum hardware.
\end{abstract}

\maketitle

\section{Introduction}
A finite representation of quantum control must specify more than a collection of admissible states. It must determine how commands act, how their inverses are implemented, what information preparation retains, and how an observable probability is obtained. Even when every selected step is close to its nominal quantum gate, several different reversible transition tables may meet the same accuracy and calibration requirements. Their long circuit histories need not agree.

This paper studies the resulting certification problem. A model is selected once and used throughout a family of experiments. Each experiment prepares the same physical state, applies a chosen gate word, and performs a terminal measurement. There are no intermediate measurements. The question is whether one controller and one preparation law can reproduce all the resulting quantum probabilities within a prescribed tolerance. Allowing an independently selected model for each circuit changes this question and can remove a discrepancy that is unavoidable for a shared model.

Finite-state and rational-grid proposals provide one motivation for considering such controllers~\cite{Palmer2025,Palmer2026}. The mathematical problem does not require a particular interpretation of those proposals, an ultrametric on state labels, a fundamental time step, or a physical information bound. We work with explicit finite frame catalogues, permutations approximating nominal commands, and a specified terminal probability rule. Frame labels remain part of the model even when distinct labels describe the same physical ray.

There are three relevant bodies of prior work. Permutation approximation of measure-preserving dynamics supplies the general construction behind reversible discretisation~\cite{Lax1971}; the matching condition is classical~\cite{Hall1935}. Memory costs of quantum contextuality concern finite simulators of sequential measurement statistics~\cite{Kleinmann2011,Trandafir2025}. Gate set tomography instead characterises quantum operations jointly from circuit data, including preparation and measurement~\cite{Nielsen2021}. The present model fixes the preparation and readout conventions and asks whether a constrained finite permutation model is compatible with a collection of probabilities. It does not perform calibration-free quantum process tomography or prove a general contextuality inequality. Circuit selection to reduce the experimental cost of gate set tomography is also studied in Ref.~\cite{Rudinger2023}.

Preparation and transformation contextuality are formulated in terms of operational equivalences~\cite{Spekkens2005}. Here the context labels specify an explicit preparation protocol; their presence alone does not establish contextuality in that operational sense.

The finite-frame constructions and repeated-control benchmark originate in a longer research archive~\cite{BruzzeseArchive}. Here the central object is the simultaneous circuit radius. We give its exact finite feasibility formulation, an analytic example of strict individual-versus-joint separation, a fresh computational reconstruction of the local benchmark, and a finite-shot interpretation. The computational package accompanying this paper contains the geometry, witness tables, verification code and figure generation needed for the results reported below. We do not rely on unavailable execution logs from the earlier archive.

The main conceptual distinction is between an \emph{existence certificate}, which exhibits a model fitting a collection, and an \emph{exclusion certificate}, which eliminates every model in the stated class. A good fit for one table is not an exclusion result; a failure of several selected tables is not an exclusion result either. Calibration, reversibility, symmetry and stationarity specify the class over which the quantifiers must range.

\section{Finite reversible frame models}\label{sec:model}
\subsection{Commands, frames and preparation}
Fix a Hilbert space $\mathbb C^d$, a finite label set $X$ of size $M$, and frame representatives $F_i\in U(d)$ for $i\in X$. The map $i\mapsto[F_i]$ need not be injective. A command alphabet $\mathcal G$ has nominal unitaries $U_g$. A stationary controller assigns one permutation $\pi_g\in S_X$ to each command. The same permutation is used whenever that command appears.

A word $w=(g_1,\ldots,g_\ell)$ is written in chronological order. Its ideal unitary and label update are
\begin{equation}
 U_w=U_{g_\ell}\cdots U_{g_1},\qquad
 \pi_w=\pi_{g_\ell}\circ\cdots\circ\pi_{g_1}.
 \label{eq:word}
\end{equation}
If inverse commands are included, their tables are fixed to $\pi_{g^{-1}}=\pi_g^{-1}$. A forward command with a self-adjoint nominal matrix need not itself be an involution. Imposing that extra identity defines a smaller model class. Reversibility here means that the retained label determines a unique predecessor under a specified command, consistent with logical reversibility~\cite{Bennett1973}. It does not by itself specify a thermodynamic implementation.

Let $\rho$ be the fixed target preparation and $0\leq E\leq I$ the terminal effect. A finite context set $K$ specifies initial labels $a_k$ and a preparation distribution $\eta\in\Pp\subseteq\Delta(K)$, where $\Pp$ is a fixed compact polytope. In context $k$, prepare the compensated seed
\begin{equation}
 \rho_k=F_{a_k}^{\dagger}\rho F_{a_k}.
 \label{eq:seed}
\end{equation}
Every branch then represents the same initial physical state. Draw the context once per run, retain it through the word, and use the same distribution for every circuit in the experiment. It is independent of the choice of circuit.

For a fixed controller $\pi$, the context and ensemble probabilities are
\begin{align}
 Q^\pi_{wk}
 &=\tr\!\left[E F_{\pi_w(a_k)}\rho_kF_{\pi_w(a_k)}^\dagger\right],
 \label{eq:context-probability}\\
 q_{\pi,\eta}(w)&=\sum_k\eta_k Q^\pi_{wk},\qquad
 p(w)=\tr(EU_w\rho U_w^\dagger).
 \label{eq:probabilities}
\end{align}
All preparations in this paper are fixed and explicitly given. The register bound $M$ counts evolving frame labels. It does not include the description of the transition tables, arbitrary precision in real-valued probabilities, or an unrestricted continuum of preparations. For a finite set of preparations, the compensated seed set is also finite.

\begin{proposition}[Consistency of compensated preparation]\label{prop:channel}
For a fixed word and preparation distribution, the ensemble rule extends to a random-unitary quantum channel
\[
 \mathcal E_w(\rho)=\sum_k\eta_k W_{wk}\rho W_{wk}^\dagger,
 \qquad W_{wk}=F_{\pi_w(a_k)}F_{a_k}^\dagger.
\]
Appending the paired inverse word returns every branch exactly to its initial physical state.
\end{proposition}
\begin{proof}
Substitute \eqref{eq:seed} into \eqref{eq:context-probability}. Each $W_{wk}$ is unitary and the weights sum to one. For reversal, $\pi_{w^{-1}}\pi_w=\id$, so the final frame is $F_{a_k}$ and cancels its preparation compensation.
\end{proof}
The channel statement does not imply $\mathcal E_{vw}=\mathcal E_v\circ\mathcal E_w$: the latter expression can implicitly redraw a context that the model retains. The primitive specification is the labelled dynamics and its preparation law. Operational descriptions of quantum memory likewise consider sequences of interventions rather than isolated input--output maps~\cite{Pollock2018}. We specify the retained frame dynamics directly and do not infer a general Markovianity criterion from the channel representation.

\subsection{Accuracy graphs and calibration}
Use the phase-optimised operator distance
\begin{equation}
 d_{\op}(V,W)=\min_{|z|=1}\|V-zW\|_{\op}.
 \label{eq:metric}
\end{equation}
The accuracy graph for command $g$ contains precisely the edges
\begin{equation}
 (i,j)\in A_g
 \quad\Longleftrightarrow\quad
 d_{\op}(F_j,U_gF_i)\leq\eps_g.
 \label{eq:graph}
\end{equation}
Prescribed calibration edges must belong to the table. Additional conditions, such as locality, commutation with a frame symmetry, or an involution requirement, are imposed explicitly. Denote the resulting nonempty finite controller class by $\Pi$.

For a transition $i\mapsto j$, the selected physical matrix is $F_jF_i^\dagger$. Thus \eqref{eq:graph} bounds its error from $U_g$. In dimension two,
\begin{equation}
 d_{\op}(V,W)^2=2-|\tr(V^\dagger W)|.
 \label{eq:trace-distance}
\end{equation}
Indeed, after removal of a common phase the two eigenvalues of $V^\dagger W$ can be centred on the shorter unit-circle arc; their maximum distance from one is determined by the trace modulus. This identity is specific to two-dimensional unitaries and is not used as a higher-dimensional formula.

\begin{proposition}[Accumulated error]\label{prop:error}
Every allowed controller and preparation prior obey
\begin{equation}
 |q_{\pi,\eta}(w)-p(w)|\leq
 \min\left\{1,\sum_{s=1}^{\ell}\eps_{g_s}\right\}.
 \label{eq:accumulated}
\end{equation}
\end{proposition}
\begin{proof}
Choose a phase for each selected step that realises its projective distance. Telescoping the product gives an operator-norm error bounded by the sum of the step errors. For any input density matrix and effect, the probability difference is bounded by the corresponding trace distance~\cite[Chap.~3]{Watrous2018}, which is at most that operator-norm error. Apply the bound branchwise and average.
\end{proof}
Consequently, improved per-gate accuracy gives convergence on every fixed finite circuit collection. It does not by itself establish uniform convergence over arbitrarily long words.

\section{Simultaneous certification}\label{sec:simultaneous}
\subsection{One model for the entire experiment}
For a finite nonempty collection $\W$, define
\begin{equation}
 R(\W)=\min_{\pi\in\Pi,\;\eta\in\Pp}
       \max_{w\in\W}|q_{\pi,\eta}(w)-p(w)|.
 \label{eq:radius}
\end{equation}
The minimum is attained because $\Pi$ is finite and $\Pp$ is compact. The inequality $R(\W)>\delta$ means that every allowed controller/prior pair fails at least one test by more than $\delta$. The failing circuit may depend on the pair. It does not assert that one circuit fails for all pairs.

The independently fitted radius is
\begin{equation}
 R_{\mathrm{ind}}(\W)=
 \max_{w\in\W}\min_{\pi\in\Pi,\;\eta\in\Pp}
 |q_{\pi,\eta}(w)-p(w)|\leq R(\W).
 \label{eq:individual}
\end{equation}
Section~\ref{sec:four} proves that the inequality can be strict even for a four-frame qubit catalogue with exactly two allowed tables. Also, enlarging $\W$ cannot decrease $R$, while enlarging $\Pi$ or $\Pp$ cannot increase it. These monotonicities make the scope of an exclusion certificate transparent.

For a fixed table, let $B^\pi_{wk}=Q^\pi_{wk}-p(w)$. Optimising the prior is the linear programme
\begin{equation}
 r_\pi(\W)=\min t\quad\text{subject to}\quad
 \eta\in\Pp,\qquad -t\mathbf1\leq B^\pi\eta\leq t\mathbf1.
 \label{eq:prior-lp}
\end{equation}
When $\Pp=\Delta(K)$, convex minimax~\cite{Sion1958} gives
\begin{equation}
 r_\pi(\W)=\max_{\|v\|_1\leq1}\min_k
               \sum_{w\in\W}v_w B^\pi_{wk}.
 \label{eq:dual}
\end{equation}
A feasible contrast $v$ can therefore certify a discrepancy for every prior of that table. A class-wide contrast must satisfy the required inequality for \emph{every} table as well. In general one cannot exchange minimisation over the discrete table class with maximisation over contrasts. Such an exchange instead corresponds to an appropriate convexified model class.

\subsection{Exact shared-table feasibility}
The table need not be enumerated before testing a collection. Introduce binary variables $x_{gij}$ with row and column sums one, prohibited edges zero, and calibration entries fixed. These variables represent each command's permutation. Inverse commands use transposed variables. Any locality or symmetry constraints tie the corresponding entries.

For every word $w$, context $k$, and step $s$, let $m^w_{ksi}$ be the probability mass at label $i$, and $f^w_{ksij}$ the mass on a transition. Impose
\begin{align}
 m^w_{k0i}&=\eta_k\mathbf1_{i=a_k},\qquad \eta\in\Pp,\nonumber\\
 \sum_j f^w_{ksij}&=m^w_{k,s-1,i},\qquad
 \sum_i f^w_{ksij}=m^w_{ksj},\label{eq:flow}\\
 0\leq f^w_{ksij}&\leq x_{g_sij},\qquad m^w_{ksi}\geq0.\nonumber
\end{align}
For $r_{kj}=\tr(EF_j\rho_kF_j^\dagger)$, require
\begin{equation}
 -\delta\leq\sum_{k,j}r_{kj}m^w_{k,|w|,j}-p(w)\leq\delta
 \quad (w\in\W).
 \label{eq:terminal}
\end{equation}
The same $x$ and $\eta$ occur in all words; only their trajectory masses differ.

\begin{theorem}[Shared-table criterion]\label{thm:flow}
Equations \eqref{eq:flow}--\eqref{eq:terminal}, together with the binary assignment and model-class constraints, are feasible if and only if $R(\W)\leq\delta$.
\end{theorem}
\begin{proof}
An allowed controller and prior supply the assignments and their actual trajectory masses. Conversely, an integral assignment has exactly one outgoing edge at each source. All other outgoing flow capacities vanish. Conservation forces the full mass at that source to follow the selected edge. Induction over steps therefore gives exactly the trajectories of one stationary controller. The terminal equations are then precisely the desired simultaneous probability bounds.
\end{proof}
For algebraic inputs this is a finite mixed-integer linear problem with algebraic coefficients. Dropping integrality generally allows mass splitting and is only a relaxation. A floating-point solver's failure to find a solution is not a proof of infeasibility.

An alternative exact search assigns transitions only when a tested history first visits an unassigned source. A new assignment must be injective and extend to a complete allowed permutation. After deleting the used rows and columns, such an extension exists exactly when the residual bipartite graph has a perfect matching~\cite{Hall1935}. Thus a failed matching check provides a sound pruning rule. Repeated visits reuse the existing assignment. The finite benchmark in Section~\ref{sec:benchmark} uses this approach. Maximum bipartite matching has efficient classical algorithms~\cite{HopcroftKarp1973}; the supplied verifier uses a simple augmenting-path implementation for its finite instances.

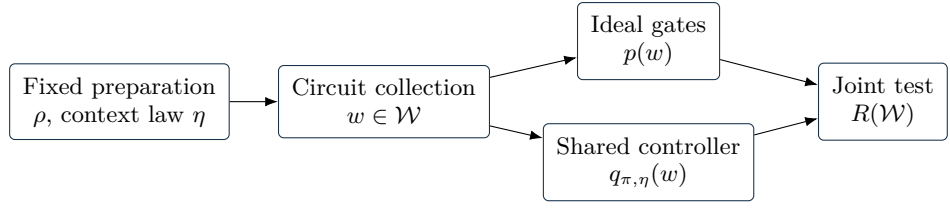
\begin{figure}[tbp]
\centering
\begin{tikzpicture}[>=Latex,box/.style={draw=linkblue,rounded corners=2pt,align=center,font=\small,inner sep=5pt}]
\node[box] (prep) at (0,0) {Fixed preparation\\$\rho$, context law $\eta$};
\node[box] (words) at (3.5,0) {Circuit collection\\$w\in\mathcal W$};
\node[box] (ideal) at (7,0.8) {Ideal gates\\$p(w)$};
\node[box] (finite) at (7,-0.8) {Shared controller\\$q_{\pi,\eta}(w)$};
\node[box] (test) at (10.1,0) {Joint test\\$R(\mathcal W)$};
\draw[->] (prep)--(words);\draw[->] (words)--(ideal);\draw[->] (words)--(finite);
\draw[->] (ideal)--(test);\draw[->] (finite)--(test);
\end{tikzpicture}
\caption{Every circuit uses the same preparation law and the same command tables. The ideal predictions and the finite-model predictions are compared as complete vectors. Different experimental runs may use different circuits, but their model parameters are shared.}
\label{fig:setup}
\end{figure}

\section{An exact separation with four frames}\label{sec:four}
The distinction between separate and joint fits can be demonstrated without a large numerical search. Let
\[
 R_y(\phi)=\begin{pmatrix}\cos(\phi/2)&-\sin(\phi/2)\\
 \sin(\phi/2)&\cos(\phi/2)\end{pmatrix},\qquad
 F_i=R_y(i\pi/2),\quad i\in\mathbb Z/4\mathbb Z.
\]
Take the nominal commands
\begin{equation}
 G=R_y(\pi/4),\qquad A=R_y(\pi/2),\qquad
 \eps_G=2\sin(\pi/16),\quad\eps_A=0.
 \label{eq:four-gates}
\end{equation}
The exact $A$ table is the cyclic shift $s(i)=i+1$. The accuracy graph for $G$ contains exactly $i\mapsto i$ and $i\mapsto i+1$: their angle errors are $\pi/4$, while the two other frame choices have projective angular error $3\pi/4$.

There are only two perfect matchings of this graph. If one source uses its shift edge, bijectivity forces the next source to shift, and hence forces every source to shift. Otherwise all sources stay fixed. Thus the complete allowed class is $\pi_G=\id$ or $\pi_G=s$; no additional translation symmetry is assumed to obtain this classification.

Prepare $\rho=\ket0\bra0$ and measure $E=\ket1\bra1$. One may start at the identity frame, or use any fixed prior over the four initial frames with compensated seeds. Both allowed tables are translations, so the endpoint probabilities are independent of that prior. For the chronological words $\W=\{(G),(G,A)\}$, put $a=(2-\sqrt2)/4$. Direct rotation gives
\begin{equation}
 p=(a,1-a),\qquad q_{\id}=(0,1/2),\qquad q_s=(1/2,1).
 \label{eq:four-vectors}
\end{equation}

\begin{theorem}[Strict individual-versus-joint separation]\label{thm:four}
For the complete four-frame class above,
\begin{equation}
 R_{\mathrm{ind}}(\W)=\frac{2-\sqrt2}{4},\qquad
 R(\W)=\frac{\sqrt2}{4}.
 \label{eq:four-radii}
\end{equation}
Each circuit has an allowed fit with error below $1/4$, but no single allowed controller fits both within $1/4$.
\end{theorem}
\begin{proof}
For the first circuit the identity table has error $a$, and the shift table has error $1/2-a=\sqrt2/4$. For the second circuit these errors are exchanged. Taking the minima separately yields $a$; taking the maximum for either shared table yields $\sqrt2/4$. The inequalities $a<1/4<\sqrt2/4$ are strict.
\end{proof}
This example has a gate-error allowance approximately $0.3902$. It proves a structural separation within an explicitly classified model; it does not claim a small-error asymptotic obstruction.

\begin{figure}[tbp]
\centering\includegraphics[width=.70\linewidth]{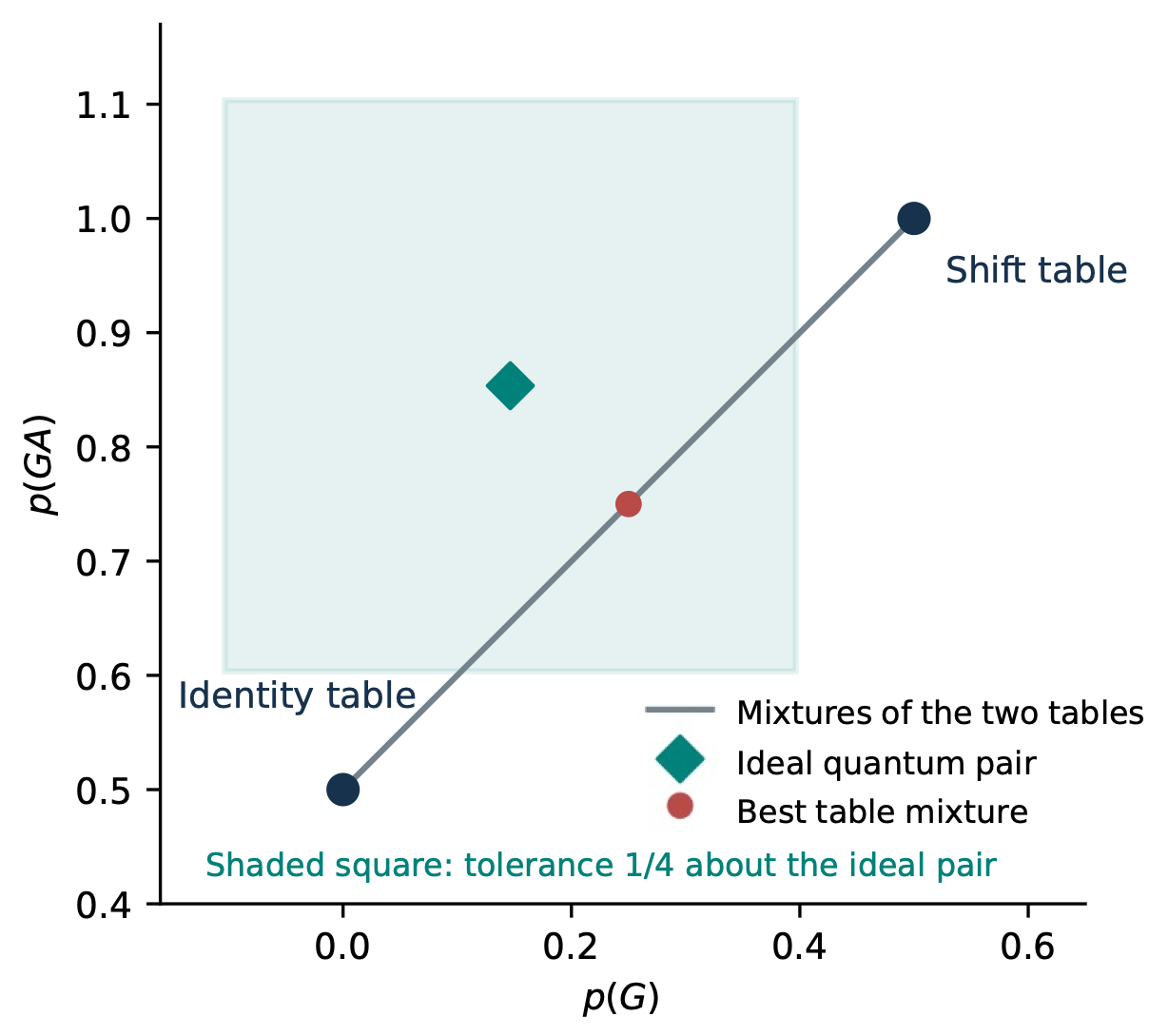}
\caption{Exact two-circuit probability vectors for Theorem~\ref{thm:four}. Each coordinate can be fitted within $1/4$ by a different table, but neither complete table vector enters the shaded tolerance square. The line segment consists of an enlarged class that randomises over the two tables between runs.}
\label{fig:four}
\end{figure}

\subsection{Why randomising the table changes the certificate}
If a fresh table may be selected with probability $\lambda$ on each run, the attainable pair becomes $(\lambda/2,1/2+\lambda/2)$. This is a different model, not an alternative choice of the compensated-frame prior. Its optimal radius is
\begin{equation}
 R_{\mathrm{mix}}(\W)=\frac{\sqrt2-1}{4},
 \label{eq:mixed-radius}
\end{equation}
achieved at $\lambda=1/2$. To prove optimality, every such pair has coordinate difference $1/2$, whereas the ideal difference is $1/\sqrt2$. Therefore its maximum coordinate error is at least half the difference of these two quantities; the midpoint attains the bound. The common linear contrast
\[
 \frac12\bigl[(q_1-p_1)-(q_2-p_2)\bigr]
 =\frac{\sqrt2-1}{4}
\]
separates even the mixed-table class from the exact target, but does not certify the larger deterministic radius. This calculation illustrates why discrete table optimisation and convex prior optimisation must be distinguished.

\section{An exact 464-frame benchmark}\label{sec:benchmark}
\subsection{Catalogue and symmetry}
Let $H=2^{-1/2}\left(\begin{smallmatrix}1&1\\1&-1\end{smallmatrix}\right)$ and $T=\diag(1,e^{i\pi/4})$. For $t=n/8$, $1\leq n\leq7$, and $\beta,\gamma\in\{0,\pi/4,\ldots,7\pi/4\}$, take
\begin{equation}
 V(t,\beta,\gamma)=
 \begin{pmatrix}
 \sqrt t&\sqrt{1-t}e^{i\beta}\\
 \sqrt{1-t}e^{i\gamma}&-\sqrt t e^{i(\beta+\gamma)}
 \end{pmatrix}.
 \label{eq:catalogue}
\end{equation}
Add the eight endpoint frames $\left(\begin{smallmatrix}0&1\\e^{i\gamma}&0\end{smallmatrix}\right)$ and the eight frames $\diag(1,-e^{i\gamma})$. These $16+7\cdot64=464$ projective frames form $\C$.

Use exact left multiplication by $T$ for the phase table $t$. For the Hadamard table $h$, require
\begin{equation}
 d_{\op}(F_{h(i)},HF_i)\leq\frac{13}{50},\qquad
 h(I)=H,\quad h(H)=I.
 \label{eq:local-constraints}
\end{equation}
Additionally require equivariance under right multiplication by $T$: if $\rho_T$ is its label permutation, then $h\rho_T=\rho_T h$. Exact left multiplication $t$ also commutes with $\rho_T$. This is an explicit model restriction, not a consequence of the accuracy bound or of finiteness.

Prepare $\ket0$ with initial contexts $T^k$ and compensated seeds $T^{-k}\ket0$, for any fixed distribution on $k=0,\ldots,7$. For every word in $H,T$, equivariance gives
\begin{equation}
 F_{\pi_w(T^k)}T^{-k}\ket0
 \sim F_{\pi_w(I)}\ket0.
 \label{eq:phase-independence}
\end{equation}
All contexts therefore have identical terminal rays and probabilities. Optimisation over the preparation prior disappears for this entire circuit family.

\subsection{The exact reduced graph}
Right multiplication by $T$ partitions the catalogue into 58 orbits of length eight. Label a frame by $(i,k)$ with an orbit representative $i\in\{0,\ldots,57\}$ and a right phase $k\in\mathbb Z/8\mathbb Z$. The repeated chronological block $(T,H)$ has table $P=ht$, and every equivariant permutation has the form
\begin{equation}
 P(i,k)=(\sigma(i),k+s_i),\qquad \sigma\in S_{58}.
 \label{eq:quotient}
\end{equation}
The representative ordering is specified in Appendix~\ref{app:labels}. The initial orbit is 1. Calibration fixes
\begin{equation}
 \sigma(1)=26,\quad s_1=5,\qquad
 \sigma(33)=1,\quad s_{33}=4.
 \label{eq:quotient-anchors}
\end{equation}
For each representative source and each of the 464 target frames, the exact edge test is
\begin{equation}
 \left|\tr(F_j^\dagger HTF_i)\right|^2
 \geq \left(2-(13/50)^2\right)^2.
 \label{eq:exact-edge}
\end{equation}
All 26,912 comparisons are performed in
$\mathbb Q(\sqrt2,\sqrt3,\sqrt5,\sqrt7,i)$. After imposing the calibration rows and columns there are exactly 174 permitted triples $(i,j,s)$. In this instance each permitted orbit pair has a unique shift. A perfect matching on these 174 edges therefore specifies a full equivariant controller. Its $H$ table is recovered as $h=Pt^{-1}$.

The readout from $\ket0$ is constant on each right-phase orbit. The two endpoint orbits have readouts 1 and 0; each interior orbit at level $n/8$ has readout $1-n/8$. Thus probability comparisons reduce to exact rational readouts and algebraic ideal probabilities.

\subsection{Sharp repeated-circuit horizon}
For $n$ repetitions of $(T,H)$, define $p_n=|\bra1(HT)^n\ket0|^2$. If
\begin{equation}
 \cos(2\theta)=-\frac{2+\sqrt2}{4},\quad 0<\theta<\frac\pi2,
 \qquad a=\frac{6-\sqrt2}{17},
 \label{eq:theta}
\end{equation}
then
\begin{equation}
 p_n=a\,[1-\cos(2n\theta)].
 \label{eq:ideal-local}
\end{equation}
Indeed $|\tr(HT)|^2=1-1/\sqrt2$, so after removal of a common phase the eigenvalues are $e^{\pm i\theta}$. The Cayley--Hamilton identity and $|(HT)_{10}|^2=1/2$ give \eqref{eq:ideal-local}. Equivalently, set $c_0=1$, $c_1=-(2+\sqrt2)/4$ and
\begin{equation}
 c_{n+1}=2c_1c_n-c_{n-1},\qquad p_n=a(1-c_n).
 \label{eq:chebyshev}
\end{equation}
This recurrence performs every probability-band comparison exactly in $\mathbb Q(\sqrt2)$.

\begin{theorem}[Certified local horizon]\label{thm:horizon}
For the full class \eqref{eq:local-constraints} with right-phase equivariance, one controller satisfies $|q_n-p_n|<13/40<1/3$ for all $0\leq n\leq192$. Every controller in this class has $|q_n-p_n|>1/3$ for at least one $1\leq n\leq193$. Both statements hold for every fixed compensated phase prior.
\end{theorem}
\begin{proof}[Proof with an exhaustive finite certificate]
The witness quotient cycle through the initial orbit is
\begin{equation}
 (1,26,56,35,40,52,17,53,18,48,43,41,50,42,49,51,33).
 \label{eq:cycle}
\end{equation}
Every edge passes \eqref{eq:exact-edge}. A complete matching extending this cycle is supplied with the paper. Its shifts sum to zero modulo eight, so the full frame returns after 17 blocks. Its probabilities repeat
\begin{equation}
 \left(0,\tfrac12,\tfrac18,\tfrac38,\tfrac38,\tfrac18,\tfrac34,
 \tfrac18,\tfrac58,\tfrac14,\tfrac14,\tfrac38,\tfrac18,
 \tfrac14,\tfrac14,\tfrac18,\tfrac12\right).
 \label{eq:sequence}
\end{equation}
Exact comparisons using \eqref{eq:chebyshev} prove the asserted $13/40$ bound through 192 and an error greater than $17/50$ at 193.

For exclusion, enumerate paths from orbit 1 with the calibration assignments fixed. At depth $n$, retain a vertex only if its readout lies in the closed interval $[p_n-1/3,p_n+1/3]$. An assigned source reuses its target. An unassigned source branches over unused graph neighbours, retaining an assignment only if it extends to a full matching. Every admissible controller traces one such branch; probability, injection and matching-extension rejections cannot remove a feasible controller. Conversely, every surviving partial assignment has a matching completion. The exhaustive search has no surviving path through depth 193, and its largest surviving depth is 192. The exact code, complete graph, witness and run counts are supplied. Equation~\eqref{eq:phase-independence} removes the prior in both parts.
\end{proof}

The reconstructed search visited 41,595 nodes, with 6,565 matching-extension rejections, 4,620 probability-band rejections and 11,854 injection rejections. These counts describe the supplied traversal, not a complexity lower bound. They differ from earlier implementation counts because this verifier checks constraints in a different order. The numerical witness maximum through 192 is approximately $0.3219128661$, and its error at 193 is approximately $0.3447811784$; the theorem uses exact inequalities rather than these rounded values.

\begin{figure}[tbp]
\centering\includegraphics[width=.96\linewidth]{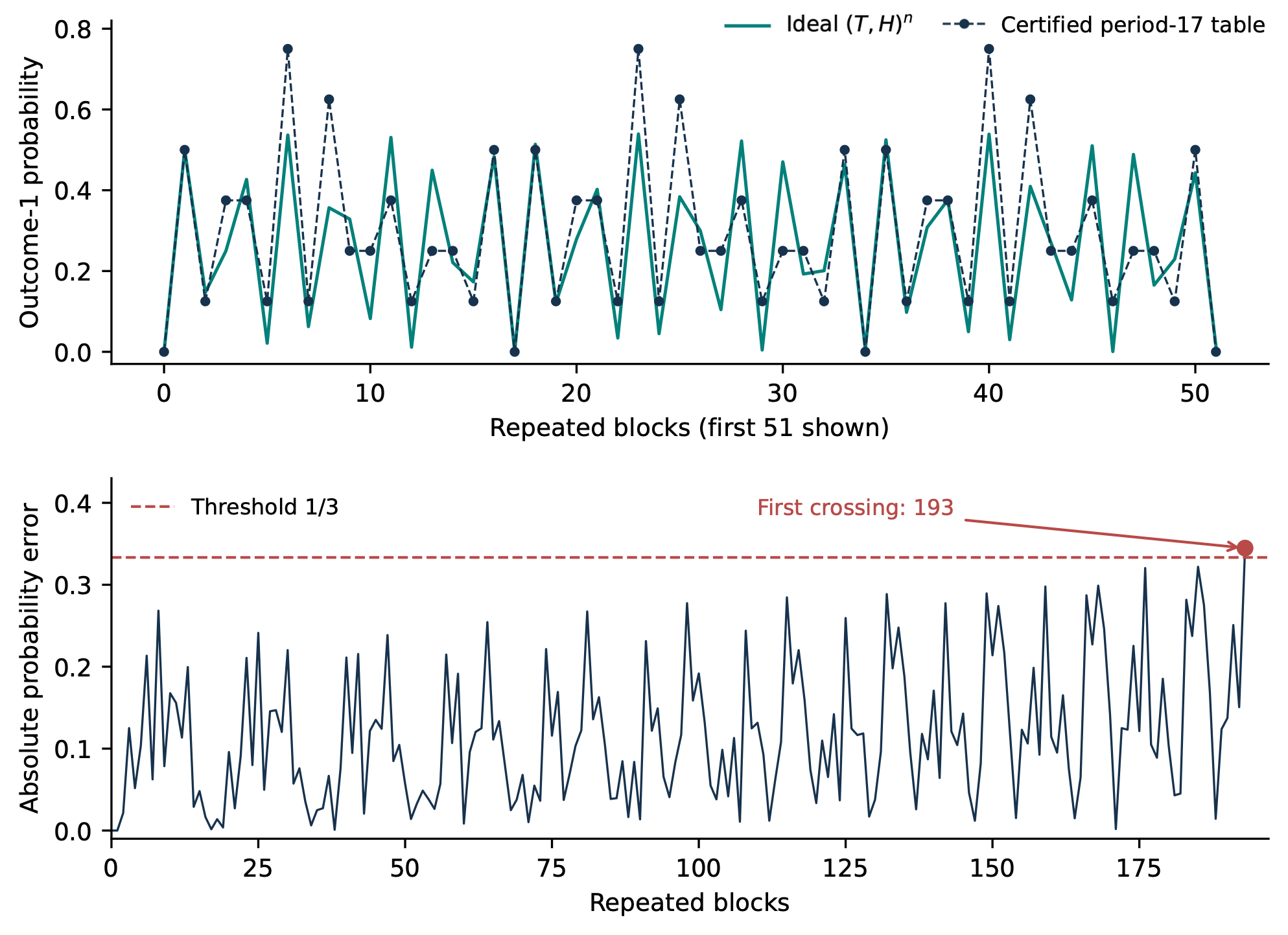}
\caption{The certified period-17 witness. Top: the first 51 repeated blocks compared with the ideal local probability. Bottom: its absolute error through block 193. This particular table crosses $1/3$ first at 193. The statement that every allowed table crosses by that horizon follows from exhaustive search, not from the plotted trajectory. Each block contains two elementary commands.}
\label{fig:horizon}
\end{figure}

\subsection{Many short circuits can still share a fit}
A repeated trajectory and a collection of short words constrain a controller differently. To make this distinction concrete, we also supply a second full quotient matching obtained by a constraint-generation search over short words.

\begin{proposition}[A simultaneous short-word witness]\label{prop:twelve}
In the same 464-frame class, one controller satisfies
\[
 |q(w)-p(w)|\leq\frac13
 \qquad\text{for every }w\in\{H,T\}^{\ell},\quad1\leq\ell\leq12.
\]
This holds for every fixed compensated phase prior. The displayed witness has maximum error approximately $0.3017766953$ over all 8,190 words.
\end{proposition}
\begin{proof}
The supplied table is a bijection on the 58 orbits, uses only permitted edges, and obeys the anchors. The unique edge shifts lift it to 464 frames. Exhaustive enumeration of all $\sum_{\ell=1}^{12}2^\ell=8190$ words, using exact field arithmetic for the ideal matrices and rational endpoint readouts, verifies each closed-band inequality. Phase independence follows from \eqref{eq:phase-independence}.
\end{proof}
The proposition is a feasibility certificate, not a proof of optimality of that table or of a twelve-gate horizon. In particular, it rules out a $1/3$ exclusion certificate consisting solely of forward $H,T$ words of length at most twelve for this class. The 192-block theorem concerns a different collection, with words extending to 386 gates. Neither result establishes that a shorter joint exclusion collection exists on this catalogue. The analytic four-frame separation demonstrates the mechanism without asserting such an improvement for the larger class.

\section{Finite-window obstruction from periodic dynamics}\label{sec:spectral}
The finite benchmark has a sharp computational horizon. A less sharp but general bound follows from the spectral difference between a permutation and an irrational quantum oscillation. Consider the ideal local curve \eqref{eq:ideal-local} and any stationary reversible frame model with at most $M$ labels and a fixed compensated preparation prior. No phase symmetry or matching choice is needed for the following argument.

First, $\theta/\pi$ is irrational. If an eigenvalue ratio $\zeta=e^{2i\theta}$ of $HT$ were a root of unity, then $|\tr(HT)|^2=2+\zeta+\zeta^{-1}$ would be an algebraic integer. But $1-1/\sqrt2$ has conjugate product $1/2$, and hence is not an algebraic integer. Define the finite constant
\begin{equation}
 \kappa_M=\max_{1\leq q\leq M}\frac{q}{|\sin(q\theta)|},\qquad
 S_D(f)=\frac1D\sum_{n=0}^{D-1}f(n)\cos(2n\theta).
 \label{eq:kappa}
\end{equation}

\begin{lemma}[Periodic correlation bound]\label{lem:correlation}
For every probability curve $f$ produced by such a model,
\begin{equation}
 |S_D(f)|\leq\frac{\kappa_M}{D}.
 \label{eq:correlation}
\end{equation}
\end{lemma}
\begin{proof}
Each initial context follows a cycle of the repeated-block permutation with length $q\leq M$. Its readout sequence $g(n)\in[0,1]$ is $q$-periodic. Split $\sum_{n<D}g(n)e^{2in\theta}$ into residue classes modulo $q$. The geometric sum in each class has modulus at most $1/|\sin(q\theta)|$, and there are at most $q$ classes. Divide by $D$, take the real part, and average with the fixed preparation weights.
\end{proof}
For $p_n=a(1-\cos(2n\theta))$, geometric sums also give
\begin{equation}
 \left|S_D(p)+\frac a2\right|\leq\frac{C}{D},\qquad
 C=\frac{a}{|\sin\theta|}+\frac{a}{2|\sin(2\theta)|}.
 \label{eq:ideal-correlation}
\end{equation}

\begin{theorem}[A common finite-window contrast]\label{thm:window}
Every controller and fixed prior in the stated finite-label class satisfy
\begin{equation}
 S_D(f-p)\geq\frac a2-\frac{\kappa_M+C}{D}.
 \label{eq:window}
\end{equation}
For any $0<\delta<a/2$, choosing
$D>(\kappa_M+C)/(a/2-\delta)$ certifies
$\max_{0\leq n<D}|f(n)-p_n|>\delta$ for every such model.
\end{theorem}
\begin{proof}
Combine the lower side of \eqref{eq:correlation} with the upper side of \eqref{eq:ideal-correlation}. Finally $|S_D(f-p)|\leq\max_{n<D}|f(n)-p_n|$.
\end{proof}
The contrast is common to the complete class, although the depth of greatest discrepancy need not be common. Small denominators $|\sin(q\theta)|$ can make the sufficient window large. No computationally efficient asymptotic onset is inferred from this bound.

There is also a necessary relation between label budget and gate accuracy. Suppose the uniform accumulated error per repeated block is $\eps$. On a frame cycle of length $q\leq M$, the product of the selected physical unitaries is projectively the identity, by telescoping the successive frames. Consequently
\[
 d_{\op}((HT)^q,I)\leq q\eps.
\]
The two eigenvalues of $(HT)^q$ are separated by $2|\sin(q\theta)|$ up to a common phase, so their distance to a common unit scalar cannot be less than $|\sin(q\theta)|$. Therefore
\begin{equation}
 \eps\geq\frac1{\kappa_M}.
 \label{eq:budget-error}
\end{equation}
This necessary condition concerns the evolving frame register. It neither specifies the full accuracy graph nor determines the earliest observable discrepancy. Proposition~\ref{prop:error} and Theorem~\ref{thm:window} describe the competing fixed-depth convergence and eventual finite-model obstruction.

\section{Finite-shot certification and a numerical experiment}\label{sec:sampling}
Exact probabilities are the starting point of certification, but sampled outcomes can be tested against the same model class. Let $m=|\W|$ and let $\widehat p_w$ be the empirical frequency from $N$ independent Bernoulli trials of circuit $w$. Runs reset the preparation according to the same fixed law. Set
\begin{equation}
 r_N=\sqrt{\frac{\log(2m/\alpha)}{2N}},\qquad
 \Q=\{(q_{\pi,\eta}(w))_{w\in\W}:\pi\in\Pi,\eta\in\Pp\}.
 \label{eq:confidence}
\end{equation}
Reject the model class when
\begin{equation}
 \operatorname{dist}_{\infty}(\widehat p,\Q)>r_N.
 \label{eq:test}
\end{equation}
Feasibility of a model within the empirical bands can be tested using Theorem~\ref{thm:flow}, with the band centres replaced by the observed frequencies.

\begin{proposition}[Sampling guarantee]\label{prop:shots}
Under any model in $\Q$, test \eqref{eq:test} rejects with probability at most $\alpha$. If the true probabilities are the ideal vector $p$ and $R(\W)>2r_N$, rejection occurs with probability at least $1-\alpha$.
\end{proposition}
\begin{proof}
Hoeffding's inequality~\cite{Hoeffding1963} and a union bound give
$\Pr(\|\widehat p-q\|_\infty>r_N)\leq2m e^{-2Nr_N^2}=\alpha$ for any true vector $q$. If $q\in\Q$, its distance bounds the distance to the class. If the true vector is $p$, the reverse triangle inequality gives
$\operatorname{dist}_\infty(\widehat p,\Q)\geq R(\W)-r_N>r_N$
on the same confidence event.
\end{proof}
For unequal shot counts one can instead use separate radii and error allocations per circuit. The selected collection and decision rule should be fixed before collecting the data used for the stated confidence level, or their adaptation must be accounted for statistically. Independent numerical design followed by fresh validation samples meets this requirement.

If actual quantum preparation, control and readout shift the true vector by at most $\nu$ in maximum norm, the separation from $\Q$ is at least $R(\W)-\nu$. Thus a sufficient condition becomes $R(\W)>\nu+2r_N$. Uncertainty in the definition of the finite-model class must instead be handled by enlarging that class. A nominal certificate is not automatically robust to either kind of uncertainty.

\subsection{Reproducible sampling study}
We applied \eqref{eq:test} to the two-circuit four-frame example, using the exact quantum vector and each of the two allowed table vectors as three separate data-generating laws. For each law and shot count $N\in\{25,50,100,200,400,800\}$, we generated 20,000 independent simulated data sets, with fixed random seed 20261003 and $\alpha=0.05$. The null class consists of the two stationary tables, not their between-run mixtures. Figure~\ref{fig:sampling} reports the resulting rejection frequencies.

At 25 and 50 shots per circuit, the observed quantum rejection frequencies were $0.71805$ and $0.99285$, respectively; at 100 shots they were $1.00000$ in the finite sample. Across all settings, the largest observed null rejection frequency was $0.0048$.

These are synthetic Bernoulli experiments. The quantum input probabilities are calculated from the gates, and the finite-model probabilities are \eqref{eq:four-vectors}; no laboratory observations are used or fitted. The study illustrates the finite-shot cost of an already proved separation. Its empirical frequencies do not replace the uniform confidence guarantee in Proposition~\ref{prop:shots}. With 20,000 repetitions, the Monte Carlo standard error of any displayed frequency is at most $1/(2\sqrt{20000})\simeq0.00354$.

\begin{figure}[tbp]
\centering\includegraphics[width=.83\linewidth]{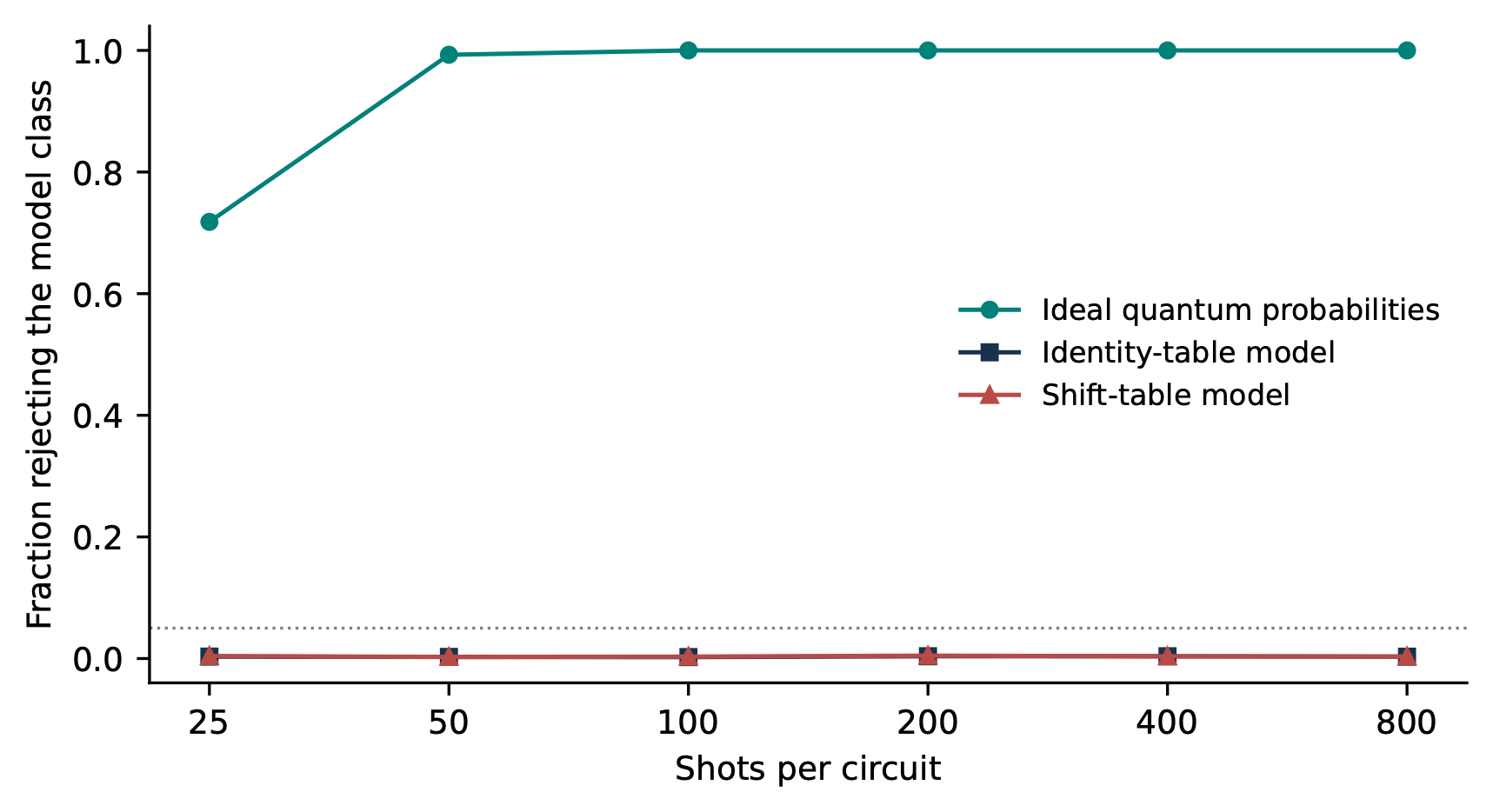}
\caption{Finite-shot numerical experiment for the four-frame example. Each point uses 20,000 independently sampled data sets at the indicated number of shots per circuit. The dotted reference is the nominal false-rejection bound $\alpha=0.05$. The two null curves overlap closely near zero. All data are simulated; the confidence guarantee is analytic.}
\label{fig:sampling}
\end{figure}

\section{Scope and conclusions}
The simultaneous circuit radius makes the shared nature of a control model explicit. The four-frame example proves that individually successful fits can be incompatible with any one stationary reversible table. The 464-frame example shows that a substantially larger collection may still admit a common fit: all forward $H,T$ words through twelve gates pass for one controller, while a repeated-block family has a certified optimal horizon of 192 blocks at threshold $1/3$. These statements have different quantifiers and different circuit resources.

The numerical endpoints depend on the specified catalogue, exact phase command, Hadamard error allowance, calibration anchors and right-phase symmetry. Removing the symmetry enlarges the class and can invalidate the exclusion part of Theorem~\ref{thm:horizon}; its explicit witness remains a witness in any larger class that contains it. Adding global gate identities, changing the readout rule, allowing stochastic transitions, or redrawing the transition table between runs also changes the certification problem. We make no claim that the four-frame or 464-frame class exhausts physically motivated finite models.

The controller register is only one resource. The definition of the catalogue, table storage, preparation context, probability precision, external command information and permitted readout operations must be counted before interpreting these results as physical memory bounds. The finite-window obstruction assumes stationary permutations and a fixed readout; it does not apply unchanged to arbitrary stochastic finite-state models or externally time-dependent controllers.

The present results provide a self-contained finite decision problem, exact existence and exclusion certificates, and a sampling rule. A natural next task is to optimise circuit collections jointly over maximum word length, number of settings, sampling cost and robustness to enlarged controller classes. Such optimisation could yield shorter certificates, but the current results do not establish an improvement over the 192-block benchmark for the full class studied here. The distinction between proved finite instances and unresolved optimisation is essential to the interpretation of the method.

\section{Verification and reproducibility}\label{sec:verification}
The manuscript \href{https://zenodo.org/records/23131105}{DOI}
and the companion code repository on
\href{https://github.com/njblottly/finite-reversible-control-source}{GitHub}
are available online.

\appendix
\section{Catalogue ordering and exact verification}\label{app:labels}
The 464-frame ordering is: eight $t=0$ endpoints, eight $t=1$ endpoints, then $(n,b,g)$ lexicographically for $1\leq n\leq7$ and $0\leq b,g<8$, with $\beta=b\pi/4$ and $\gamma=g\pi/4$. Thus $I$ is frame 12 and $H$ is frame 208. Right-phase orbit representatives have original indices
\[
 (0,8),\quad\text{followed by}\quad
 \{16+64u+g:0\leq u\leq6,\;0\leq g<8\}.
\]
The representative of orbit 1 is $Z$; $I$ differs by a right phase. Since the input is $\ket0$ and the tables are equivariant, either represents the same probability initialisation.

For the repeated-block quotient, the exact phase table has orbit action
\[
 t(0)=0,\quad t(1)=1,\quad
 t(2+8u+g)=2+8u+(g+1\bmod8).
\]
For an arbitrary chronological word, the quotient Hadamard action is $h=\sigma\circ t^{-1}$. These rules, the shift on each edge, and the full quotient witness determine all 464 transitions without an implicit numerical convention.

The arithmetic implementation represents the field
$\mathbb Q(\sqrt2,\sqrt3,\sqrt5,\sqrt7,i)$ in its 32-element product basis, with rational coefficients. Products reduce repeated square roots to integers and use $i^2=-1$. Conjugation changes the sign of the imaginary basis components. For a real expression, signs are certified by rational bounds on its positive square roots. With denominator $2^b$, integer square root gives
\[
 \frac{\lfloor\sqrt{m2^{2b}}\rfloor}{2^b}
 \leq\sqrt m\leq
 \frac{\lfloor\sqrt{m2^{2b}}\rfloor+1}{2^b},
\]
with equality handled exactly for perfect squares. Coefficient signs determine which endpoints contribute to the lower and upper sums. The precision is increased when necessary; unresolved signs raise an error rather than silently choosing a branch. Algebraic zero is detected in the product basis.

The geometry verifier classifies every representative-source/target pair using \eqref{eq:exact-edge}. The horizon verifier rebuilds the witness matching, checks its lift, and exhausts every permitted history to depth 193 with exact probability bands. A separate enumeration checks every word of length at most twelve for the short-word witness. The search uses a standard augmenting-path matching algorithm over integer labels. It never treats a floating-point optimisation status as an exclusion certificate.

\begin{table}[tbp]
\caption{Scope of the reproduced finite computations. Search counts are implementation-specific. The short-word result certifies feasibility, while the repeated-block result includes exhaustive exclusion at the next depth.}
\label{tab:verification}
\begin{ruledtabular}
\begin{tabular}{lr}
Verified object & Count or value\\
\colrule
Distinct projective catalogue frames &464\\
Right-phase orbits &58\\
Exact geometric comparisons &26,912\\
Calibrated quotient edges &174\\
Largest feasible repeated-block horizon at $1/3$ &192\\
Exhaustive search nodes in the supplied verifier &41,595\\
Forward $H,T$ words checked through length twelve &8,190\\
Sampling repetitions per law and shot count &20,000\\
\end{tabular}
\end{ruledtabular}
\end{table}

\FloatBarrier
\section{Existence of reversible approximations}\label{app:matching}
The framework assumes a nonempty controller class. A classical matching construction explains why accurate reversible tables can exist on appropriately chosen refining catalogues. The statement below is included to separate that existence question from the certification problem.

\begin{lemma}[Equal-measure matching construction]
Let a compact group have a bi-invariant metric and normalised Haar measure. Suppose a measurable partition $P_1,\ldots,P_M$, modulo null sets, has equal measures $1/M$ and centres $F_i$ satisfying $d(x,F_i)\leq r$ almost everywhere in $P_i$. For every left translation $U$, there is a permutation $\pi$ such that
\[
 \max_i d(F_{\pi(i)},UF_i)\leq2r.
\]
\end{lemma}
\begin{proof}
Join $i$ to $j$ if $UP_i\cap P_j$ has positive measure. The union of cells indexed by any input set $I$ has translated measure $|I|/M$ and is contained, modulo null sets, in its neighbouring output cells. Therefore the neighbour set has at least $|I|$ elements. Hall's theorem gives a perfect matching. A point in each matched intersection, away from exceptional null sets, bounds the centre distance by $2r$.
\end{proof}
This is an application of the established permutation-approximation method~\cite{Lax1971,Hall1935}, not a claim that arbitrary pointwise rounding is reversible. It also does not guarantee every extra calibration or symmetry condition on a chosen finite catalogue.

One exact calibration edge can be imposed at a controlled cost. If $UF_{i_0}$ is a catalogue frame and a permutation approximates $U$ with error $b$, swap the image of $i_0$ with that of the row currently mapped to $UF_{i_0}$. The new anchor is exact, and the triangle inequality bounds the changed second row by $2b$. Repeating for consistent distinct anchors preserves earlier anchors, with the corresponding repeated factor in this elementary bound. This supplies an existence mechanism; the sharp finite tests in the main text use the actual graph rather than this coarse bound.

\bibliography{references}

@misc{Palmer2025,
  author        = {Palmer, T.},
  title         = {Rational Quantum Mechanics: Testing Quantum Theory with Quantum Computers},
  year          = {2026},
  eprint        = {2510.02877},
  archivePrefix = {arXiv},
  note          = {arXiv:2510.02877v3},
  url           = {https://arxiv.org/abs/2510.02877v3}
}

@misc{Palmer2026,
  author        = {Palmer, T.},
  title         = {Solving the Mysteries of Quantum Mechanics: Why Nature Abhors a Continuum},
  year          = {2026},
  eprint        = {2602.16382},
  archivePrefix = {arXiv},
  note          = {arXiv:2602.16382v1},
  url           = {https://arxiv.org/abs/2602.16382v1}
}

@article{Lax1971,
  author  = {Lax, P. D.},
  title   = {Approximation of measure preserving transformations},
  journal = {Communications on Pure and Applied Mathematics},
  volume  = {24},
  pages   = {133--135},
  year    = {1971},
  doi     = {10.1002/cpa.3160240204}
}

@article{Hall1935,
  author  = {Hall, P.},
  title   = {On Representatives of Subsets},
  journal = {Journal of the London Mathematical Society},
  volume  = {s1-10},
  pages   = {26--30},
  year    = {1935},
  doi     = {10.1112/jlms/s1-10.37.26}
}

@article{Kleinmann2011,
  author  = {Kleinmann, M. and G{\"u}hne, O. and Portillo, J. R. and Larsson, J.-{\AA}. and Cabello, A.},
  title   = {Memory cost of quantum contextuality},
  journal = {New Journal of Physics},
  volume  = {13},
  pages   = {113011},
  year    = {2011},
  doi     = {10.1088/1367-2630/13/11/113011}
}

@misc{Trandafir2025,
  author        = {Trandafir, S. and Kelleher, C. and Cabello, A.},
  title         = {Memory cost of quantum contextuality with {Pauli} observables},
  year          = {2025},
  eprint        = {2506.06869},
  archivePrefix = {arXiv},
  note          = {arXiv:2506.06869},
  url           = {https://arxiv.org/abs/2506.06869}
}

@article{Nielsen2021,
  author  = {Nielsen, E. and Gamble, J. K. and Rudinger, K. and Scholten, T. and Young, K. and Blume-Kohout, R.},
  title   = {Gate Set Tomography},
  journal = {Quantum},
  volume  = {5},
  pages   = {557},
  year    = {2021},
  doi     = {10.22331/q-2021-10-05-557}
}

@misc{Rudinger2023,
  author        = {Rudinger, K. M. and Ostrove, C. I. and Seritan, S. K. and Grace, M. D. and Nielsen, E. and Blume-Kohout, R. J. and Young, K. C.},
  title         = {Two-Qubit Gate Set Tomography with Fewer Circuits},
  year          = {2023},
  eprint        = {2307.15767},
  archivePrefix = {arXiv},
  note          = {arXiv:2307.15767v2},
  url           = {https://arxiv.org/abs/2307.15767v2}
}

@article{Spekkens2005,
  author  = {Spekkens, R. W.},
  title   = {Contextuality for preparations, transformations, and unsharp measurements},
  journal = {Physical Review A},
  volume  = {71},
  pages   = {052108},
  year    = {2005},
  doi     = {10.1103/PhysRevA.71.052108}
}

@unpublished{BruzzeseArchive,
  author = {Bruzzese, J.},
  title  = {Gate Composition and Admissible Control on Rational Quantum Grids: Towards discrete dynamics for {RaQM}},
  year   = {2026},
  note   = {Consolidated research archive, version 0.18 (15 September 2026). The version used here is the author's supplied manuscript. Earlier project archive: \url{https://doi.org/10.5281/zenodo.22313144}}
}

@article{Bennett1973,
  author  = {Bennett, C. H.},
  title   = {Logical Reversibility of Computation},
  journal = {IBM Journal of Research and Development},
  volume  = {17},
  pages   = {525--532},
  year    = {1973},
  doi     = {10.1147/rd.176.0525}
}

@article{Pollock2018,
  author  = {Pollock, F. A. and Rodr{\'i}guez-Rosario, C. and Frauenheim, T. and Paternostro, M. and Modi, K.},
  title   = {Operational {Markov} Condition for Quantum Processes},
  journal = {Physical Review Letters},
  volume  = {120},
  pages   = {040405},
  year    = {2018},
  doi     = {10.1103/PhysRevLett.120.040405}
}

@book{Watrous2018,
  author    = {Watrous, J.},
  title     = {The Theory of Quantum Information},
  publisher = {Cambridge University Press},
  year      = {2018},
  note      = {Author's manuscript and errata: \url{https://cs.uwaterloo.ca/~watrous/TQI/}},
  url       = {https://cs.uwaterloo.ca/~watrous/TQI/}
}

@article{Sion1958,
  author  = {Sion, M.},
  title   = {On general minimax theorems},
  journal = {Pacific Journal of Mathematics},
  volume  = {8},
  pages   = {171--176},
  year    = {1958},
  doi     = {10.2140/pjm.1958.8.171}
}

@article{HopcroftKarp1973,
  author  = {Hopcroft, J. E. and Karp, R. M.},
  title   = {An {$n^{5/2}$} Algorithm for Maximum Matchings in Bipartite Graphs},
  journal = {SIAM Journal on Computing},
  volume  = {2},
  pages   = {225--231},
  year    = {1973},
  doi     = {10.1137/0202019}
}

@article{Hoeffding1963,
  author  = {Hoeffding, W.},
  title   = {Probability Inequalities for Sums of Bounded Random Variables},
  journal = {Journal of the American Statistical Association},
  volume  = {58},
  pages   = {13--30},
  year    = {1963},
  doi     = {10.1080/01621459.1963.10500830}
}
\end{document}